\documentclass{article}
\usepackage[
    pdfauthor={Luigi Alfonsi, Hyungrok Kim},
    pdftitle={On coefficients of operator product expansions for quantum field theories with ordinary, holomorphic, and topological spacetime dimensions},
    pdfsubject={hep-th,math-ph},
    pdfkeywords={operator product expansion, twist, sheaf cohomology, derived functions},
    breaklinks=true,
    final
]{hyperref}
\usepackage{amsmath,amssymb,amsthm,mleftright,booktabs,orcidlink,subcaption}
\usepackage[german,french,russian,ukrainian,latin,british]{babel}
\usepackage{CJKutf8}
\usepackage{cleveref}

\newtheorem{theorem}{Theorem}
\newtheorem{lemma}[theorem]{Lemma}
\theoremstyle{definition}

\title{Higher operator product expansions in twisted non-Lorentzian field theories}
\author{Luigi Alfonsi\orcidlink{0000-0001-5231-2354}\and Hyungrok Kim~(\begin{CJK*}{UTF8}{bsmi}金炯錄\end{CJK*})\orcidlink{0000-0001-7909-4510}\\[1em]
\texttt{\{l.alfonsi,h.kim2\}@herts.ac.uk}\\[1em]
Centre for Mathematics and Theoretical Physics Research\\Department of Mathematics and Theoretical Physics\\University of Hertfordshire\\Hatfield, Hertfordshire \textsc{al10~9ab}, United Kingdom
}
\usepackage[final,spacing]{microtype}
\usepackage[osf]{newpxtext}
\usepackage[vvarbb,slantedGreek]{newpxmath}
\renewcommand\Delta\Deltaup
\renewcommand\Gamma\Gammaup
\renewcommand\pi\piup

\begin{document}
\maketitle
\begin{abstract}
Twisting theories with super-Poincaré symmetry yields topological--holomorphic field theories, whose operator product expansions of local operators can have as coefficients not only ordinary functions but also `derived' functions with nonzero cohomological degree, described by certain elements of sheaf cohomology.
Twisting non-Lorentzian supersymmetric theories yields richer phenomena, where topological and holomorphic spacetime coordinates coexist with ordinary (that is, neither topological nor holomorphic) spacetime coordinates.
We propose a sheaf-cohomological description of such OPE coefficients for quantum field theories with arbitrary numbers of topological, holomorphic, and/or ordinary spacetime dimensions, and identify necessary and sufficient conditions for the proposed OPE coefficient space to contain nontrivial `derived' functions; in particular, theories with one topological spacetime dimension and multiple ordinary spacetime dimensions provide a smooth analogue of the (holomorphic) raviolo.
We then demonstrate that two-point functions in topological twists of free supersymmetric non-Lorentzian field theories appearing in the literature define nontrivial derived functions consistent with our sheaf-theoretic construction.
\end{abstract}
\newpage
\tableofcontents

\section{Introduction and summary}
Super-Poincaré-invariant quantum field theories can be twisted into topological--holomorphic field theories \cite{Costello:2011np,Elliott:2020ecf}, which have attracted much recent interest \cite{Saberi:2019fkq,Bomans:2023mkd,Gaiotto:2024gii}. In such theories, spacetime is locally of the form \(\mathbb R^m\times\mathbb C^n\), such that all physical quantities, such as correlation functions and operator-product expansions, are invariant under deformations of spacetime coordinates along the \(\mathbb R^m\) directions and under anti-holomorophic deformations along the \(\mathbb C^n\) directions.

In a one-dimensional holomorphic field theory, where spacetime is \(\mathbb C\), one expects OPEs to be of the form
\begin{equation}
    O_i(z)O_j(0) \sim \sum_k C_{ijk}(z)O_k(0),
\end{equation}
where the coefficient functions \(C_{ijk}(z)\) are functions holomorphic on \(\mathbb C\setminus\{0\}\). On the other hand, for a higher-dimensional holomorphic theory, living on \(\mathbb C^n\) for \(n>1\), a naïve OPE of the form
\begin{equation}
    O_i(z)O_j(0) \sim \sum_k C_{ijk}(z)O_k(0)
\end{equation}
where now \(C_{ijk}(z)\) are functions holomorphic on \(\mathbb C^n\setminus\{0\}\) is not tenable because of the Hartogs extension theorem \cite{hartogs}, which implies that all such functions are entire (i.e.\ holomorphic on all of \(\mathbb C^n\)).
Instead, it has been argued \cite{williamsthesis} that in such cases the `coefficients' of the OPE should rather be valued in the \emph{sheaf cohomology} \(\operatorname H^\bullet(\mathbb C^n\setminus\{0\},\mathcal O)\) of the sheaf \(\mathcal O\) of holomorphic functions on \(\mathbb C^n\setminus\{0\}\). 
In derived algebraic geometry, it is a general philosophy that elements of sheaf cohomology should be interpreted as a generalisation of the notion of sections of a sheaf (in this case, of the notion of holomorphic functions); in particular, the zeroth cohomology \(\operatorname H^0(\mathbb C^n\setminus\{0\},\mathcal O)\) reduces to the notion of ordinary holomorphic functions. (Furthermore, reassuringly, in the complex one-dimensional case \(n=1\) the sheaf cohomology is concentrated in degree zero, so that higher phenomena do not occur.)
Similarly, topological--holomorphic theories on \(\mathbb R\times\mathbb C\) are described byraviolo vertex algebras \cite{Garner:2023zqn,Garner:2023zko,Alfonsi:2024qdr}, where the relevant sheaf of functions is the sheaf of functions holomorphic along \(\mathbb C\) and constant along \(\mathbb R\), which similarly has nontrivial higher cohomology.

The possibilities are richer if we keep supersymmetry while relaxing the assumption of Lorentz symmetry.
In that case, twisting such theories in general produces a blend of topological, holomorphic, and ordinary (that is, neither topological nor holomorphic) spacetime directions; locally, these twisted theories then live on a spacetime of the form \(\mathbb R^m\times\mathbb C^n\times\mathbb R^p\), where the theory is topological along \(\mathbb R^m\) but ordinary along \(\mathbb R^p\).
In this case, it is natural to conjecture that the OPE coefficients are elements of
\begin{equation}
    \operatorname H^\bullet((\mathbb R^m\times\mathbb C^n\times\mathbb R^p)\setminus\{0\},\mathcal O),
\end{equation}
where \(\mathcal O\) is the sheaf of smooth functions that are constant along \(\mathbb R^m\) and holomorphic along \(\mathbb C^n\).
We compute this sheaf cohomology using the Čech-to-derived spectral sequence (see \cref{thm:main}), and find that it generically has non-vanishing higher cohomology. In particular, when \((m,n,p)=(1,0,p)\), then we obtain a smooth analogue of the holomorphic raviolo.

We then test this on examples of free supersymmetric non-Lorentzian field theories found in the literature \cite{Redigolo:2011bv,Arav:2019tqm} and explicitly show that the Green functions explicitly define higher derived functions in accordance with the sheaf-theoretic computation.

\paragraph{Organisation of this paper.}
This paper is organised as follows. In \cref{sec:main}, we state and formulate the main theorem on sheaf cohomology of the sheaf of constant--holomorphic--smooth functions as above, which apply to two-point operator-product expansions. In \cref{sec:examples}, we then discuss various physical cases and the interpretation of the higher cohomology in various cases. In \cref{sec:theory-examples}, we discuss two-point functions in supersymmetric twists of non-Lorentzian theories discussed in \cite{Redigolo:2011bv,Arav:2019tqm} and shows that they define derived functions (sheaf-cohomology classes) consistent with the main theorem.
Finally, in \cref{sec:discussion} we briefly describe new features expected to appear the multi-point case and possible generalisations.

\section{Local geometries}\label{sec:main}
Consider a quantum field theory in \(m+2n+p\) spacetime dimensions, of which \(m\) are topological, \(2n\) pair up into \(n\) complex holomorphic dimensions, and \(p\) are ordinary (neither topological nor holomorphic). That is, spacetime locally appears as
\begin{equation}\mathbb R^m\times\mathbb C^n\times\mathbb R^p.\end{equation}
Let us coordinatise the above as \((x_1,\dotsc,x_m,z_1,\dotsc,z_n,y_1,\dotsc,y_p)\), such that OPEs and correlation functions do not depend on \(x_1,\dotsc,x_m\) or \(\bar z_1,\dotsc,\bar z_n\).

On any open subset \(U\subset\mathbb R^m\times\mathbb C^n\times\mathbb R^p\), define
\begin{equation}\label{eq:sheaf}
    \mathcal O(U) = \left\{
        f\in\mathcal C^\infty(U,\mathbb C)
        \,\middle|\,
        0 = 
        \frac{\partial f}{\partial x_1} = \dotsb
        =
        \frac{\partial f}{\partial x_m}
        =\frac{\partial f}{\partial\bar z_1}
        =\dotsb=\frac{\partial f}{\partial\bar z_n}
    \right\},
\end{equation}
that is, the space of complex-valued smooth functions locally constant along \(x_1,\dotsc,x_m\) and holomorphic along \(z_1,\dotsc,z_n\). It is clear that the above defines a sheaf on \(\mathbb R^m\times\mathbb C^n\times\mathbb R^p\) equipped with the ordinary (i.e.\ analytic) topology. Thus it makes sense to speak of its sheaf cohomology.

We consider the local cohomology with support at the origin of \(\mathbb C^n\):
\begin{equation}
    \operatorname{H}^\bullet_{\{0\}}(\mathbb{C}^n,\mathcal{O}) = \mathbf{R}^\bullet\Gamma_{\{0\}}(\mathbb{C}^n,\mathcal{O}),
\end{equation}
where \(\Gamma_{\{0\}}(\mathbb{C}^n,-)\) denotes the functor of sections with support at the origin, and \(\mathbf R\) denotes the right derived functor. In particular, notice that we have
\begin{equation}
    \operatorname H^n_{\{0\}}(\mathbb{C}^n,\mathcal{O})
    \cong\begin{cases}
        \mathbb{C} &\text{if \(n=0\)}\\
        \mathcal{O}(\mathbb{C}\setminus\{0\})/\mathcal{O}(\mathbb{C}) &\text{if \(n=1\)}\\
        \textcolor{blue}{\mathcal O(\mathbb C^n)}
        &\text{if \(n\geq 2\)}.
    \end{cases}
\end{equation}
    (More precisely, in the case \(n\ge2\), an element of \(\operatorname H^n_{\{0\}}(\mathbb{C}^n,\mathcal{O})\cong\operatorname H^{n-1}(\mathbb{C}^n\setminus\{0\},\mathcal{O})\) is represented by a function of the form
    \begin{equation}
        z_1^{-1}\dotsm z_n^{-1}f(1/z_1,\dotsc,1/z_n),
    \end{equation}
    where \(f\) is an entire function on \(\mathbb C^n\). This identification, however, depends on the choice of coordinates.)
Now, we define the function spaces
\begin{subequations}\label{eq:function_spaces}
\begin{align}
    Z_{n,p} &= \mathcal O(\mathbb C^n)\otimes\mathcal C^\infty(\mathbb R^p,\mathbb C) \\
    \tilde Z_{n,p} &= \mathcal O(\mathbb C^n)\otimes\mathcal C^\infty(\mathbb R^p\setminus\{0\},\mathbb C) \\
    Y_{n,p} &= \operatorname H_{\{0\}}^n(\mathbb C^n,\mathcal O)\otimes\mathcal C^\infty(\mathbb R^p,\mathbb C) \\
    \tilde Y_{n,p} &= \operatorname H_{\{0\}}^n(\mathbb C^n,\mathcal O)\otimes\mathcal C^\infty(\mathbb R^p\setminus\{0\},\mathbb C) ,
\end{align}
\end{subequations}
where in the above \(\mathcal O(\mathbb C^n)\) is the space of holomorphic functions on \(\mathbb C^n\), and \(\mathcal C^\infty(\mathbb R^p,\mathbb C)\) is the space of complex-valued smooth functions on \(\mathbb R^p\), and \(\otimes\) is the topological tensor product of complex nuclear Fréchet spaces\footnote{For general Fréchet topological vector spaces, one has a family of possible topological tensor products, of which the `smallest' is the injective tensor product and the `largest' is the projective tensor product; for nuclear Fréchet spaces, these two agree, so that it makes sense to speak of \emph{the} topological tensor product.}; these are all nuclear Fréchet spaces (concentrated in degree zero --- note that \(\operatorname H(\mathbb C^n\setminus\{0\},\mathcal O)\) is concentrated in degrees \(0\) and \(n-1\), and if \(n>1\) then \(\operatorname H^0(\mathbb C^n\setminus\{0\},\mathcal O)=Z_{n,0}\)).

In particular, \(Z_{0,p}=\mathcal C^\infty(\mathbb R^p,\mathbb C)\), and \(Y_{1,p}=(\mathcal O(\mathbb C\setminus\{0\})/\mathcal O(\mathbb C))\otimes\mathcal C^\infty(\mathbb R^p,\mathbb C)\); when \(p=0\), then \(\tilde Z_{n,0}=\tilde Y_{n,0}=0\) since \(\mathcal C^\infty(\varnothing)=0\) is the zero-dimensional vector space.

\begin{theorem}\label{thm:main}
Consider \((\mathbb R^m\times\mathbb C^n\times\mathbb R^p)\setminus\{0\}\) equipped with the sheaf \(\mathcal O\) of complex-valued smooth functions that are constant along \(\mathbb R^m\) and holomorphic along \(\mathbb C^n\). Then, as a graded complex vector space, its sheaf cohomology is
\begin{equation}\label{eq:cohomology}
    \operatorname H((\mathbb R^m\times\mathbb C^n\times\mathbb R^p)\setminus\{0\},\mathcal O)
    \cong\begin{cases}
        Z_{n,p}\oplus(\tilde Y_{n,p}/Y_{n,p})[-m-n] &\text{if \(p\ge1\)}\\
        Z_{n,p}\oplus Y_{n,p}[1-m-n]&\text{if \(m+n>0\) and \(p=0\)}\\
        0 &\text{if \(m=n=p=0\)}.
    \end{cases}
\end{equation}
\end{theorem}
The proof of this theorem, a direct computation in sheaf cohomology, is given in \cref{app:appendix}.

Given this computation, we make the following physical hypothesis: the possible coefficients in an OPE in a quantum field theory on \(\mathbb R^m\times\mathbb C^n\times\mathbb R^p\) with \(m\) topological, \(n\) holomorphic, and \(p\) ordinary spacetime directions are given by the sheaf cohomology \eqref{eq:cohomology}. In particular, we see that higher structure exists (i.e.\ we must use derived functions rather than ordinary functions) precisely when \(m+n>0\) and \(p>0\) or when \(m+n>1\) and \(p=0\). 
In the next section, we will compare this hypothesis with physical expectations.

Note that the result of the above computation is not always\footnote{
    In some cases, there is no need to take this limit. For instance, in the smooth-raviolo case (\cref{ssec:smooth-raviolo}), the quotient \(\mathcal C^\infty(\mathbb R^p\setminus\{0\})/\mathcal C^\infty(\mathbb R^p)\) is already fully local.
} local in the sense that it assumes that spacetime is globally --- not just locally --- of the form \(\mathbb R^m\times\mathbb C^n\times\mathbb R^p\), and will not be applicable to spacetimes with nontrivial topology. In addition, for non-conformal field theories, operator product expansions may only exist in a limit where the two insertion points become arbitrarily close. In that case, one may take the following limit at some fixed point \(x\in\mathbb R^m\times\mathbb C^n\times\mathbb R^p\):
\begin{equation}\label{eq:germs}
    G_{m,n,p}\coloneqq\varinjlim_{x\ni U}\operatorname H(U\setminus\{x\},\mathcal O),
\end{equation}
where the limit ranges over open neighbourhoods of \(x\).
(This clearly does not depend on the choice of \(x\).)
Geometrically, this corresponds to taking a pro-object in the category of locally ringed spaces \cite{SB_1958-1960__5__369_0} (reviewed in e.g.\ \cite{zbMATH02193967}).

\section{Special cases, and the smooth raviolo}\label{sec:examples}
\subsection{Holomorphic field theory}
For the case \((m,n,p)=(0,1,0)\), corresponding to a one-dimensional holomorphic field theory or equivalently a two-dimensional chiral conformal field theory, the space \eqref{eq:cohomology} is the space of holomorphic functions on \(\mathbb C^\times\), allowing an isolated singularity (either a pole or an essential singularity) at the origin.
Apart from the analytic convergence condition (which one may remove by taking a formal completion), and allowing essential singularities, this is the OPE that one expects from a one-dimensional holomorphic theory, namely
\begin{equation}
    O_i(z)O_j(0)=\sum_k\sum_{l\in\mathbb Z}c_{ijk}^{(l)}z^lO_k(0).
\end{equation}

On the other hand, consider the case where \((m,n,p)=(0,n,0)\) with \(n\ge2\), such as holomorphic Chern--Simons theory. Let us recall the following well-known lemma:
\begin{lemma}\label{lem:sheaf-coho-result}
    Let \(n\ge2\). The sheaf cohomology of the complex manifold \(\mathbb C^n\setminus\{0\}\) is
    \begin{equation}\label{eq:sheaf-coho-result}
        \operatorname H(\mathbb C^n\setminus\{0\},\mathcal O)=
            \mathcal O(\mathbb C^n) \oplus \left(\mathcal O((\mathbb C^\times)^n)/X_n\right)[1-n],
    \end{equation}
    where \(X_n\subset\mathcal O((\mathbb C^\times)^n)\) is the space of holomorphic functions \(f\in\mathcal O((\mathbb C^\times)^n)\) that can be represented as
    \begin{equation}
        f = f_1+\dotsb+f_n
    \end{equation}
    where
    \begin{equation}
        f_i\in\mathcal O\left(\left((\mathbb C^\times)^{i-1}\times\mathbb C\times(\mathbb C^\times)^{n-i}\right)\right).
    \end{equation}
\end{lemma}\begin{proof}
    
    The complex manifold \(\mathbb C^n\setminus\{0\}\) is covered by the open sets
    \begin{equation}
        \mathcal U=\left\{U_i = \mathbb C^{i-1}\times(\mathbb C\setminus\{0\})\times\mathbb C^{n-i}\middle|i\in\{1,\dotsc,n\}\right\},
    \end{equation}
    which are domains of holomorphy and hence Stein.
    Their intersections are also Stein (because they are also domains of holomorphy), so Cartan's theorem~B implies that \(\mathcal U\) is a Leray cover, and Čech cohomology with respect to \(\mathcal U\) computes the sheaf cohomology.

    Hartogs' extension theorem shows that the degree-zero cohomology is \(\mathcal O(\mathbb C^n)\). In higher degrees, Laurent expansion decomposes the \v{C}ech complex according to the variables appearing with negative powers, and the Čech differential preserves this decomposition. Only terms with negative powers in all \(n\) variables give a contribution of \(\mathcal O((\mathbb C^\times)^n)/X_n\) in degree \(n-1\). Thus, the cohomology is trivial except in the top and bottom degrees, as in \eqref{eq:sheaf-coho-result}.
\end{proof}

The appearance of this higher-degree cohomology has been discussed in \cite{Saberi:2019fkq,Bomans:2023mkd,Gaiotto:2024gii}; concretely, such OPE coefficients may be extracted using the Bochner--Martinelli kernel that represents the higher cohomology.

\subsection{Ordinary quantum field theory}
Consider the case where \((m,n,p)=(0,0,p)\), which describes a \(p\)-dimensional ordinary quantum field theory.\footnote{We always work in Euclidean signature, i.e.\ with imaginary time. In Lorentzian theories, OPEs have singularities along lightcones.} (This includes, for example, \(p\)-dimensional conformal field theories.)
In that case \eqref{eq:cohomology} simply reduces to spaces of smooth functions on \(\mathbb R^p\setminus\{0\}\).
That is, the OPE has the ordinary form
\begin{equation}
    O_i(y)O_j(0)=\sum_kc_{ijk}(y)O_k(0)
\end{equation}
with arbitrary smooth functions \(c_{ijk}\), as expected. (This has used only translation symmetry. Of course, if one has additional symmetry such as rotational symmetry or conformal symmetry, one can further constrain the form of the functions \(c_{ijk}\).)

\subsection{Topological field theory}
Consider the case \((m,n,p)=(m,0,0)\), which corresponds to an \(m\)-dimensional topological field theory. In this case, \eqref{eq:cohomology} is
\begin{equation}
    \operatorname H(\mathbb R^m\setminus\{0\},\underline{\mathbb C}) = \mathbb C\oplus\mathbb C[1-m].\label{eq:tqft-result}
\end{equation}
Physically, it is known \cite{Lurie:2009keu,Scheimbauer:2014zty} that an \(m\)-dimensional TQFT corresponds to \(E_m\)-algebras valued in e.g.\ chain complexes, and these in turn correspond to locally constant factorisation algebras \cite{Costello:2023knl}.
In a topological field theory in \(m\) (real) spacetime dimensions, the operations on local observables are governed by the little \(m\)-discs operad \(E_m\). Over characteristic zero, the dg-operad associated to the topological operad \(E_m\) is formal \cite{Kontsevich:1999eg,2008arXiv0808.0457L}, i.e.\ quasi-isomorphic to its cohomology. For \(m\ge2\), the cohomology is the \(m\)-Poisson operad \(P_m\). That is, the higher operations of arity \(k\) are given either by the cohomology of the configuration space \(\operatorname{Conf}_k(\mathbb R^m)\), or equivalently by the \(k\)-ary operations of a \(P_{m,\infty}\)-algebra (using Koszul duality of operads). In particular, since \(\operatorname{Conf}_2(\mathbb R^m)\simeq\mathbb S^{m-1}\), the space of binary operations is given by
\begin{equation}
    \operatorname H^\bullet_\mathrm{sing}(\operatorname{Conf}_2(\mathbb R^m);\mathbb R)\cong
    \operatorname H^\bullet_\mathrm{sing}(\mathbb S^{m-1};\mathbb R)\cong\mathbb R\oplus\mathbb R[1-m].
\end{equation}
This is the same space that appears in the definition of the binary part of an \(E_m\)-algebra, so that we get a \(P_m\)-algebra as expected. So there are two ways to multiply local operators: an operator \(\bullet\) by juxtaposition, and a Poisson bracket \([-,-]\) of degree \(1-m\).
That is, the `OPE' in a TQFT is of the form
\begin{equation}
    O_i O_j = \sum_k c_{ijk} O_k + d_{ijk} O_k,
\end{equation}
where \(c_{ijk}\in\mathbb C\), the structure constant of the ordinary multiplication \(\bullet\), is a number of degree \(0\), while and \(d_{ijk}\in\mathbb C[1-m]\), the structure constant of the Poisson bracket \([-,-]\), is a number of degree \(m-1\), where \(c_{ijk}+d_{ijk}\) is an element of \(G_{m,0,0}=\mathbb C\oplus\mathbb C[1-m]\) as expected. These two products correspond, respectively, to the following two situations: the product \(\bullet\) simply corresponds to a juxtaposition of the two local operators, while the Poisson bracket \([-,-]\) corresponds to taking a descendent (via the descent equations) of one of the two operators, resulting in an \((m-1)\)-form-valued operator, and integrating it in a cycle around the other operator (\cref{fig:tqft-diagram}). Note that these operations are the simplest cases of correlation functions corresponding to operators inserted along more general linked submanifolds that arise in e.g.\ the discussion of higher-form symmetries \cite{Gaiotto:2014kfa}, cf.\ the review \cite{Bhardwaj:2023kri}.
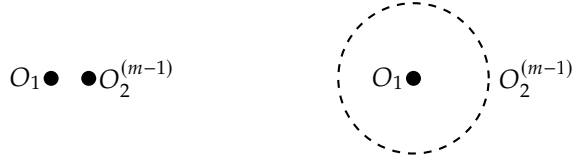
\begin{figure}
\centering
\subcaptionbox{Multiplication by juxtaposition}{
\begin{tikzpicture}
    \node[fill=black, circle, inner sep=2pt] (bullet) at (0, 0) {};
    \node at (0, 0) [left] {$O_1$};
    \node[fill=black, circle, inner sep=2pt] (bullet) at (0.5, 0) {};
    \node at (0.5, 0) [right] {$O_2$};
    \draw[thick, white] (0, 0) circle(2);
\end{tikzpicture}}
\qquad
\subcaptionbox{Poisson bracket \([O_1,O_2]\)
}{\begin{tikzpicture}
    \node[fill=black, circle, inner sep=2pt] (bullet) at (0, 0) {};
    \node at (0, 0) [left] {$O_1$};
    \draw[thick, dashed] (0, 0) circle(1);
    \node at (1, 0) [right] {$O^{(m-1)}_2$};
    \draw[thick, white] (0, 0) circle(2);
\end{tikzpicture}}
\caption{The two ways of `multiplying' local operators in a topological quantum field theory: in the first case, one simply juxtaposes the two operators; in the second case, one takes a \((m-1)\)-form-valued descendent \(O_2^{(m-1)}\) of one of the local operators \(O_2\) and wraps it around the other.}\label{fig:tqft-diagram}
\end{figure}

The case \(m=1\), that is, for a one-dimensional TQFT, is slightly different. In that case, the dg-operad associated to the \(E_1\) topological operad is known to be quasi-isomorphic (over characteristic zero) to the operad \(\operatorname{Ass}\) of associative algebras, in which case we still have two products (corresponding to the two different orderings); this still agrees with \eqref{eq:tqft-result}.

\subsection{Holomorphic raviolo}
Consider the case where \((m,n,p)=(1,n,0)\). The case \((m,n,p)=(1,1,0)\) is the case of the raviolo, as discussed in \cite{Garner:2023zqn,Garner:2023zko,Alfonsi:2024qdr}; such theories may be obtained by, for instance, twists of three-dimensional theories \cite{Elliott:2020ecf}.

In this case, the space \eqref{eq:cohomology} becomes
\begin{equation}
    \operatorname H((\mathbb R\times\mathbb C^n)\setminus\{0\},\mathcal O)\cong\mathcal O(\mathbb C^n) \oplus  \frac{\mathcal O((\mathbb C^\times)^n)}{X_n}[-n],
\end{equation}
where \(X_n\) is as in \eqref{eq:sheaf-coho-result}.
As noted in \cite{Garner:2023zqn}, when \(n=1\) this may be modelled by the (formal) complex disc with a doubled origin, the so-called raviolo, which is a non-Hausdorff complex manifold (or, algebraically, a non-separated scheme over \(\mathbb C\)).
Ignoring details such as analytic convergence, for \(n=1\), this is the same sheaf that appears in \cite{Garner:2023zqn}. As expected, the sheaf cohomology, alias the space of derived functions, contains both regular functions (in degree zero) and singular (derived) functions (in degree \(n\)).

\subsection{Smooth raviolo}\label{ssec:smooth-raviolo}
Let us consider the case \((m,n,p)=(1,0,p)\), such that \(M=\mathbb R\times\mathbb R^p\), with \(p\ge1\). In this case, the space \eqref{eq:cohomology} becomes
\begin{equation}
    \operatorname H(\mathbb R^1\times\mathbb R^p\setminus \{0\},\mathcal{O}) \cong \mathcal C^\infty(\mathbb R^p,\mathbb C) \oplus \frac{\mathcal C^\infty(\mathbb R^p\setminus\{0\},\mathbb C)}{\mathcal C^\infty(\mathbb R^p,\mathbb C)}[-1].
\end{equation}
This is the same as the sheaf cohomology \(\operatorname H(\mathbb R^p_\mathrm{dbl},\mathcal C^\infty)\) of the sheaf of smooth functions on the non-Hausdorff manifold \(\mathbb R^p_\mathrm{dbl}\) that is \(\mathbb R^p\) with a doubled origin --- which we may call a `smooth raviolo'.\footnote{It may be surprising to hear that the sheaf of smooth functions, which is ordinarily fine, can have nontrivial higher cohomology; but the acyclicity of the sheaf of smooth functions on a manifold requires the assumptions of paracompactness and Hausdorffness, the latter of which fails here.}

Now, the degree zero part represents the ordinary part of the operator product expansion; physically, it is required to be nonsingular because two operator insertions can always be moved arbitrarily far apart by moving along the topological direction,
\begin{equation}
    O_1(x,\vec y)O_2(0,\vec0)\implies O_1(x+\Delta x,\vec y)O_2(0,\vec0),
\end{equation}
so that divergences cannot occur in the limit \(\vec y\to0\). On the other hand, the higher part \( \left(\mathcal C^\infty(\mathbb R^p\setminus\{0\},\mathbb C)/\mathcal C^\infty(\mathbb R^p,\mathbb C)\right)[-1]\) is allowed to be singular --- indeed, it is the space of singularity classes of smooth functions on \(\mathbb R^p\setminus\{0\}\), modulo those which extend smoothly through the origin.

We can describe the degree-\(1\) coefficient space in a more invariant way using relative contours. Suppose that \(p\geq1\). Consider a submersion \(M\coloneqq\mathbb R^{p+1}\twoheadrightarrow\mathbb R^p\).
Let us call the fibre direction `time' and the transverse directions `space' for convenience. Fix an insertion point \(\bullet\in M\) so that
\(0=\pi(\bullet)\). We choose the coordinate \(x\in \mathbb{R}\) on the fibres.
Let us call \(\operatorname d_\pi\) the exterior derivative along the fibres and \(\Omega_\pi^1\) the sheaf of complex-valued leafwise one-forms. So, one has \(\operatorname d_\pi f=(\partial_xf)\,\mathrm{d}x\) for any smooth function \(f\) on \(M\setminus\{\bullet\}\), and any \(\alpha\in\Omega_\pi^1(M\setminus\{\bullet\})\) may be written as \(\alpha=a(x,y)\, \mathrm{d}x\) for some \(a\in\mathcal C^\infty(M\setminus\{\bullet\},\mathbb C)\).
The sheaf \(\mathcal O\) of functions locally constant along the fibres has the fine resolution
\begin{equation}
    0\longrightarrow\mathcal O_{M\setminus\{\bullet\}} \longrightarrow \mathcal C^\infty_{M\setminus\{\bullet\}} \xrightarrow{d_\pi}\Omega_\pi^1 \longrightarrow 0.
\end{equation}
and thus,
\begin{equation}\label{eq:smooth_raviolo_leafwise}
    \operatorname H^1(M\setminus\{\bullet\},\mathcal O) \cong \frac{\Omega_\pi^1(M\setminus\{\bullet\})}{d_\pi\mathcal C^\infty(M\setminus\{\bullet\},\mathbb C)}.
\end{equation}
Now, choosing any two smooth transverse sections \(s_\pm\colon \mathbb{R}^p\rightarrow M\) such that
\begin{equation}
    x_-(y)<0<x_+(y)
\end{equation}
where \(s_\pm(y)=(x_\pm(y),y)\), we can define the oriented fibre interval \(\gamma_y\) from \(s_-(y)\) to \(s_+(y)\) for any \(y\in \mathbb{R}^p\setminus\{0\}\). These intervals avoid the insertion point, even if the one over \(y_0\) would pass through it.
For any \(\alpha\in\Omega_\pi^1(M\setminus\{\bullet\})\), define
\begin{equation}\label{eq:smooth_raviolo_relative_integral}
    I_{s_-}^{s_+}(\alpha; y) \coloneqq\int_{\gamma_y}\alpha = \int_{x_-(y)}^{x_+(y)}a(x,y)\,\mathrm{d}x.
\end{equation}
This is smooth on \(y\in\mathbb{R}^p\setminus\{0\}\) and it does not need to extend smoothly at \(y=0\). In fact, one can see that the singularity class at \(y=0\) is given by
\begin{equation}
    \tilde{Y}_{0,p}/Y_{0,p} = \frac{\mathcal C^\infty(\mathbb{R}^p\setminus\{0\},\mathbb C)}{\mathcal C^\infty(\mathbb{R}^p,\mathbb C)},
\end{equation}
The singularity class of \(I_{s_-}^{s_+}(\alpha; -)\) depends only on the leafwise cohomology class of \(\alpha\). Indeed, adding an exact leafwise \(1\)-form by \(\alpha+\mathrm{d}_\pi \beta\) changes the integral by
\begin{equation}\label{eq:smooth_raviolo_endpoint_terms}
    I_{s_-}^{s_+}(d_\pi \beta; y) =  \beta(s_+(y))-\beta(s_-(y)).
\end{equation}
Both \(\beta\circ s_\pm\) extend smoothly at \(y=0\) because the sections \(s_\pm\) avoid \(\bullet\), and thus their difference vanishes in \(\tilde{Y}_{0,p}/Y_{0,p}\). Moreover, changing either \(s_-\) or \(s_+\) while keeping its intersection with the fibre at \(y=0\) on the same side of \(\bullet\) adds an integral over a fibre interval away from \(\bullet\), and such a contribution is smooth at \(y=0\). Thus the integral above defines a map independent of all these choices,
\begin{equation}\label{eq:smooth_raviolo_relative_map}
    \begin{aligned}
    \mathcal I\colon \operatorname H^1(M\setminus\{\bullet\},\mathcal O) &\to \tilde{Y}_{0,p}/Y_{0,p}\\
    [\alpha]&\mapsto [I_{s_-}^{s_+}(\alpha;-)],
    \end{aligned}
\end{equation}
which is also an isomorphism. To see this, we can cover \(M\setminus \{\bullet\}\) by two open sets \(U_\pm \coloneqq M\setminus\{(x,0)\,|\mp x\geq 0\}\). Each chart can be contracted along the fibres to the corresponding section \(s_\pm\), and their intersection is \((\mathbb{R}^p\setminus\{0\})\times\mathbb R\). The charts and their intersection are \(\mathcal O\)-acyclic, so their Čech complex is
\begin{equation}
\begin{aligned}
    \mathcal C^\infty(\mathbb{R}^p,\mathbb C)
    \oplus
    \mathcal C^\infty(\mathbb{R}^p,\mathbb C)
    &\to
    \mathcal C^\infty(\mathbb{R}^p\setminus \{0\},\mathbb C)\\
    (f_+,f_-)&\mapsto
        \left.f_-\right|_{\mathbb{R}^p\setminus \{0\}}- \left.f_+\right|_{\mathbb{R}^p\setminus \{0\}}.
\end{aligned}
\end{equation}
More explicitly, if a leafwise one-form \(\alpha\) has primitives \(\beta_\pm\) on \(U_\pm\), their difference on the intersection is \(\beta_--\beta_+=I_{s_-}^{s_+}(\alpha;y)\).
Thus \(\mathcal I\) is precisely the identification of the leafwise cohomology class with its \v{C}ech class.

\begin{figure}
\begin{equation}
\begin{tikzpicture}[xscale=0.5,yscale=0.5,baseline={([yshift=-0.5ex](1,1.25))}]
\draw[dashed, thick] (2,2.5) circle(2.7);
\draw (0,0) -- (0,5);
\draw (1,0) -- (1,5);
\draw (2,0) -- (2,2);
\draw (2,3) -- (2,5);
\node [circle, fill=black, inner sep=2pt] at (2,2.5) {};
\draw (2,2.5) node [anchor=south west] {\(\!O_2\)};
\draw (0,4.5) node [anchor=east] {\(O_1^{(1)}\)};
\draw (3,0) -- (3,5);
\draw (4,0) -- (4,5);

\draw (-0.5,0) node [anchor=east] {$x$};
\draw (0,-0.5) node [anchor=north] {$\vec y$};
\draw[->, dashed] (-0.5,-0.5) -- (1,-0.5);
\draw[->, dashed] (-0.5,-0.5) -- (-0.5,1);
\end{tikzpicture}
~\Longrightarrow~
\begin{tikzpicture}[xscale=0.5,yscale=0.5,baseline={([yshift=-0.5ex](1,1.25))}]
\draw[dashed, thick] (2,2.5) circle(1.3);
\draw[dashed, thick] (3.3,2.5) -- (5,2.5);
\draw[dashed, thick] (-1,2.5) -- (0.7,2.5);
\draw (0,0) -- (0,5);
\draw (1,0) -- (1,5);
\draw (2,0) -- (2,2);
\draw (2,3) -- (2,5);
\node [circle, fill=black, inner sep=2pt] at (2,2.5) {};
\draw (2,2.5) node [anchor=south west] {\(\!O_2\)};
\draw (0,3.5) node [anchor=east] {\(O_1^{(1)}\)};
\draw (3,0) -- (3,5);
\draw (4,0) -- (4,5);

\draw (-0.5,0) node [anchor=east] {$x$};
\draw (0,-0.5) node [anchor=north] {$\vec y$};
\draw[->, dashed] (-0.5,-0.5) -- (1,-0.5);
\draw[->, dashed] (-0.5,-0.5) -- (-0.5,1);
\end{tikzpicture}
\end{equation}
\caption{
Smooth raviolo.
}
\label{fig:smooth-raviolo}
\end{figure}

To interpret this construction in a cohomological field theory,
assume that the complex \((\mathcal A,Q)\) of local observables comes with a translation homotopy \(G_x\) satisfying
\begin{equation}
    Q^2=0,\qquad [Q,G_x]=\partial_x,\qquad |Q|=1,\quad |G_x|=-1.
\end{equation}
For a \(Q\)-closed operator \(O_1\), recall that its leafwise first
descendant is
\begin{equation}
    O_1^{(1)}=(G_xO_1)\,\mathrm{d}x, \qquad QO_1^{(1)}=\mathrm{d}_\pi O_1.
\end{equation}
Let \(O_2\) be another \(Q\)-closed local operator, inserted at
\(\bullet\) (see \cref{fig:smooth-raviolo}). Assume that the separated-point products and their
descendants are smooth families valued in a fixed output observable
complex \(\mathcal A\), that the following integrals are defined
there, and that the descent identities hold without anomalous
contact terms away from coincidence. Define the degree \(-1\)-bracket
\begin{equation}\label{eq:smooth_raviolo_relative_descendant}
    \langle\!\langle [O_1],[O_2] \rangle\!\rangle \coloneqq
    \left[\int_{\gamma_y}O_1^{(1)}(z)\,O_2(\bullet)\right].
\end{equation}
In fact, the descent equation gives
\begin{equation}\label{eq:smooth_raviolo_relative_ward}
    Q\int_{\gamma_y}O_1^{(1)}(z)\,O_2(\bullet) =
    O_1(s_+(y))O_2(\bullet)-O_1(s_-(y))O_2(\bullet).
\end{equation}
The endpoint sections are separated from \(\bullet\), including
at \(y=0\), so the right-hand side extends smoothly at \(y=0\). Moreover, replacing either operator by a \(Q\)-cohomologous representative also leaves the class unchanged. 
Thus, the relative descendant is \(Q\)-closed in the quotient
complex
\begin{equation}(\tilde{Y}_{0,p}/Y_{0,p})\otimes\mathcal A\cong{\mathcal C^\infty(\mathbb{R}^p\setminus\{0\},\mathcal A)}/{\mathcal C^\infty(\mathbb{R}^p,\mathcal A)}\end{equation}
with \(Q\) acting on the \(\mathcal A\)-valued families.
Finally, the class \(\langle\!\langle [O_1],[O_2] \rangle\!\rangle\in \operatorname H_Q(\tilde{Y}_{0,p}/Y_{0,p}(\mathcal A))\) is independent of the transverse sections.

\section{Topological--holomorphic--ordinary twists of non-Lorentzian supersymmetric field theories}\label{sec:theory-examples}

This paper treats derived functions that we conjecture should be the coefficients of operator product expansion in quantum field theories with mixed topological, holomorphic, and ordinary spacetime directions. This raises the question of whether there exist interesting examples of such theories.
Twists of theories with super-Poincar\'e symmetry typically yield topological--holomorphic theories \cite{Elliott:2020ecf}, that is, without ordinary spacetime dimensions, which are well known.
To go further, we turn to constructions from non-Lorentzian field theories: we now describe a general family of free supersymmetric theories that includes the free Lifshitz model of Arav--Oz--Raviv-Moshe~\cite{Arav:2019tqm} and the free homogeneous \(c=0\) models of Redigolo~\cite[§3.2]{Redigolo:2011bv} as special cases. These models admit supersymmetric twists, and their (derived) two-point functions correspond to nontrivial sheaf-cohomology classes in accordance with \cref{thm:main}.

Prior to twisting, for the scalar field \(\phi\), the OPE with its complex conjugate is schematically
\begin{equation}
    \phi(x)\bar\phi(0)\sim G(x)\mathbb1+:\!\phi\bar\phi\!:(0)
      + \mathcal O(\partial\phi,\partial\bar\phi).
\end{equation}
After twisting, we obtain a topological--holomorphic--ordinary theory on \(\mathbb R^m\times\mathbb C^n\times\mathbb R^p\).
The term \(G(x)\mathbb1\) can no longer appear in the ordinary (non-derived) OPE since it depends on the full spacetime coordinate \(x\) nontrivially. However, we show that \(G(x)\) defines a nontrivial sheaf-cohomology class \([\alpha_P]\) \eqref{eq:general-free-cohomology-class} according to \cref{thm:main}, so that the derived OPE of the twisted theory retains part of the information corresponding to this term (namely, the sheaf-cohomology class of \(G(x)\)): the \(\phi\phi\) (or, depending on the dimension, \(\phi\chi\), where \(\chi\) is the superpartner of \(\phi\)) OPE of the twisted theory contains a degree \(m+n\) coefficient \([\alpha_P]\) in addition to the usual degree-zero coefficients.

\subsection{Field-theoretic setup}
Let spacetime be split as
\begin{equation}
    \mathbb R^{\hat d}\times\mathbb R^{\bar d}
    \qquad(1\le\hat d\le4,\;1\le\bar d)
\end{equation}
with coordinates \(\hat x^{\hat\mu}\) in the hatted directions and
\(y^i\) in the barred directions.  We consider four real supercharges
\((\mathsf q_\alpha,\mathsf q^\dagger_{\dot\alpha})\), with
\(\alpha,\dot\alpha\in\{1,2\}\), satisfying
\begin{equation}\label{eq:general-free-superalgebra}
    \{\mathsf q_\alpha,\mathsf q_\beta\}=0,
    \qquad
    \{\mathsf q^\dagger_{\dot\alpha},
       \mathsf q^\dagger_{\dot\beta}\}=0,
    \qquad
    \{\mathsf q_\alpha,\mathsf q^\dagger_{\dot\beta}\}
      =-2\sigma^{\hat\mu}_{\alpha\dot\beta}
       \mathsf p_{\hat\mu}.
\end{equation}
The matrices \(\sigma^{\hat\mu}\) are the usual Pauli matrices \(\sigma^\mu_{\alpha\dot\beta}\), but where \(\mu\) is restricted to the values \(\{0,\dotsc,\hat d-1\}\).
The supercharges transform as spinors under \(\operatorname{Spin}(1,\hat d-1)\) (that is, the subalgebra spanned by \(\mathsf q_\alpha,\mathsf q^\dagger_{\dot\alpha},\mathsf p_{\hat\mu},\mathsf M_{\hat\mu\hat\nu}\) is just the dimensional restriction of the
four-dimensional \(\mathcal N=1\) super-Poincaré algebra).
On the other hand, the supercharges are scalars under rotations of the barred coordinates:\footnote{
    This means that, when \(\hat d+\bar d=4\), our supersymmetry algebra does \emph{not} reduce to the four-dimensional non-Lorentzian supersymmetry algebra of \cite[§3.2]{Redigolo:2011bv}
    since there the supercharges transform under nontrivial representations of \(\operatorname{Spin}(\bar d)\).
    This will affect the possible interaction terms of the theory, but this will not be important for the free theory;
    furthermore, this does not affect the details of the twisting procedure even for interacting theories. 
}
\begin{equation}\label{eq:general-free-transverse-scalars}
    [\mathsf M_{ij},\mathsf q_\alpha]
       =[\mathsf M_{ij},\mathsf q^\dagger_{\dot\alpha}]=0,
\end{equation}
where \(\mathsf M_{ij}\) generate \(\operatorname{Spin}(\bar d)\).
In particular, when \(\hat d=1\), these indices
are internal indices and \eqref{eq:general-free-superalgebra} reduces to
the time-domain supersymmetry algebra of \cite{Arav:2019tqm}.

The algebra \eqref{eq:general-free-superalgebra} is represented on a
superspace whose bosonic coordinates are \((\hat x^{\hat\mu},y^i)\) and whose
Grassmann coordinates are \((\theta^\alpha,\bar\theta^{\dot\alpha})\), similar to \cite{Redigolo:2011bv,Arav:2019tqm} and to the usual \(\mathcal N=1\) super-Poincaré superspace (in particular, we can speak of chiral superfields).
The supercharges differentiate only in the \(\hat x\) and \((\theta,\bar\theta\)
directions, so that \(\partial/\partial y^i\) commutes with the supercharges.  In particular, if \(\Phi\) is
a chiral superfield, then so is \(P\Phi\) for every differential operator \(P\) built from \(\partial/\partial y^i\).

Consider one chiral superfield
\begin{equation}
    \Phi=\phi+\sqrt2\theta\chi+\theta\theta F.
\end{equation}
The component fields are scalars under \(\operatorname{Spin}(\bar d)\),
although \(\chi\) retains its spinor index for the hatted symmetry.
Then a supersymmetric action can be written as usual,
\begin{equation}
    S = \int\mathrm d^{\hat d}\hat x\,\mathrm d^{\bar d}y\,\mathrm d^4\theta\,K(\Phi,\bar\Phi)+\left(
    \int\mathrm d^{\hat d}\hat x\,\mathrm d^{\bar d}y\,\mathrm d^2\theta\,W(\Phi)+\text{h.c.}\right),
\end{equation}
for a Kähler potential \(K\) and a superpotential \(W\). For a free theory, we can take them to be
\begin{align}\label{eq:general-free-KW}
    K&=\bar\Phi\Phi,
    &
    W&=\frac12\Phi P\Phi,
\end{align}
where \(P=P(-\Delta_y)\) is a non-constant differential operator that is strictly positive\footnote{i.e.\ the real polynomial \(P(s)\) only takes positive values whenever \(s\) is nonnegative} and also a polynomial of \(-\Delta_y=-\sum_i\delta^{ij}\partial^2/\partial y^i\partial y^j\), the Laplacian along \(y^i\) only.
After Wick rotation of the time coordinate \(\hat x^0\) and elimination of the
auxiliary field \(F\), the bosonic part of the action \(S\) is
\begin{equation}\label{eq:general-free-bosonic-action}
    S_{\mathrm{bos},\mathrm{Eucl}}
      =\int \mathrm d^{\hat d}\hat x\,\mathrm d^{\bar d}y\,
       \bar\phi\bigl(-\Delta_{\hat x}+P^2\bigr)\phi.
\end{equation}
For \(\hat d=1\) and \(P(-\Delta_y)=\mu-g\Delta_y\) (with \(\mu,g>0\)), this specialises to the
free \(z=2\) supersymmetric model of
Arav--Oz--Raviv-Moshe~\cite[(2.29)]{Arav:2019tqm}.  When
\(\hat d+\bar d=4\), this reduces
to the same action as the free homogeneous model of
\cite[§3.2]{Redigolo:2011bv} where the parameter \(c\) defined therein is set to zero.\footnote{
    That is, up to the fact that the action of \(\operatorname{Spin}(\bar d)\) is defined differently for fermions, as noted above.
    However, the free action in \cite[§3.2]{Redigolo:2011bv} takes the exact same form as \eqref{eq:general-free-bosonic-action}.}

\subsection{The twist}
We now prove the existence of the translation homotopies. The nature of the twist depends on \(\hat d\in\{1,2,3,4\}\).
In the case \(\hat d=4\), then \((\mathsf q_\alpha,\mathsf q^\dagger_{\dot\alpha},\mathsf p_{\hat\mu},\mathsf M_{\hat\mu\hat\nu})\) forms an algebra isomorphic to the four-dimensional \(\mathcal N=1\) super-Poincaré algebra, so we can simply do the equivalent of the holomorphic twist of four-dimensional \(\mathcal N=1\) super-Poincaré theories.
That is, 
the hatted momentum generator \(\mathsf p_{\hat\mu}=\mathsf p_{\alpha\dot\beta}\) may be written in spinor indices as
\begin{equation}
    (\mathsf p_{\alpha\dot\beta})
      =\begin{pmatrix}
          \mathsf p_{\bar z^1} & \mathsf p_{z^2}\\
          \mathsf p_{\bar z^2} & \mathsf p_{z^1}
        \end{pmatrix}.
\end{equation}
selected by \(Q=\mathsf q^\dagger_{\dot1}\).
Choosing
\begin{equation}
    Q=\mathsf q^\dagger_{\dot1},
\end{equation}
then we have \(Q^2=0\) as well as the homotopies
\begin{equation}\label{eq:general-free-four-dimensional-homotopies}
    [Q,\tfrac12\mathsf q_1]=\mathsf p_{\bar z^1},
    \qquad
    [Q,\tfrac12\mathsf q_2]=\mathsf p_{\bar z^2}.
\end{equation}
This identifies
\begin{equation}
    \mathbb R^{\hat d}\times\mathbb R^{\bar d}=\mathbb C^2\times\mathbb R^p
\end{equation}
where \(p=\bar d\). Hence the twist is a complex--ordinary theory with \((m,n,p)=(0,2,\bar d)\).

For \(\hat d<4\), choose the dimensional reduction and the complex
coordinates compatibly with \(Q\), so that the images of the
\(Q\)-exact translations in
\eqref{eq:general-free-four-dimensional-homotopies} are precisely the
translations displayed below. This yields:
\begin{equation}\label{eq:general-free-twist-cases}
\begin{array}{c|c|c|c}
  \hat d & \text{hatted space} & (m,n) & \text{nonzero homotopies}\\
  \hline
  1 & \mathbb R_{x^1}
    & (1,0) & [Q,G_1]=\partial_{x^1}\\
  2 & \mathbb C_{z^1}
    & (0,1) & [Q,G_{\bar1}]=\partial_{\bar z^1}\\
  3 & \mathbb C_{z^1}\times\mathbb R_{x^3}
    & (1,1) & [Q,G_{\bar1}]=\partial_{\bar z^1},\quad
              [Q,G_1]=\partial_{x^3}\\
  4 & \mathbb C_{z^1}\times\mathbb C_{z^2}
    & (0,2) & [Q,G_{\bar i}]=\partial_{\bar z^i},\quad i=1,2.
\end{array}
\end{equation}
Here each \(G\) is the image, under dimensional reduction, of the
appropriate nonzero multiple of \(\mathsf q_1\) or \(\mathsf q_2\). They satisfy
\begin{equation}\label{eq:general-free-proved-homotopies}
\begin{aligned}
    Q^2&=0,
    & [Q,G_a]&=\partial_{x^a},
    & [Q,G_{\bar\imath}]&=\partial_{\bar z^i},\\
    [G_a,G_b]&=0,
    & [G_a,G_{\bar\imath}]&=0,
    & [G_{\bar\imath},G_{\bar\jmath}]&=0,
\end{aligned}
\end{equation}
with
\begin{equation}\label{eq:general-free-proved-decomposition}
    \mathbb R^{\hat d}\cong\mathbb R^m\times\mathbb C^n,
    \qquad
    m=(\hat d\bmod 2)\in\{0,1\},
    \qquad
    n=\left\lfloor\tfrac12\hat d\right\rfloor.
\end{equation}
The twisted theory therefore has
\begin{equation}
    (m,n,p)=\left(
       \hat d\bmod 2,
       \left\lfloor\tfrac12\hat d\right\rfloor,
       \bar d
    \right)
\end{equation}
topological, holomorphic, and ordinary directions, respectively.

By chirality, \(Q\phi=0\).  Introduce the partial differential
\begin{equation}
    \mathrm d'=\mathrm d_x+\bar\partial_z
\end{equation}
with respect to the decomposition
\eqref{eq:general-free-proved-decomposition}.  The descendants of \(\phi\)
are
\begin{equation}\label{eq:general-free-descendants}
    \phi^{(r)}
      =\frac1{r!}
       \sum_{A_1,\ldots,A_r}
       (G_{A_1}\dotsm G_{A_r}\phi)\,
       e^{A_1}\wedge\dotsm\wedge e^{A_r},
\end{equation}
where
\begin{equation}
    (G_A,e^A)\in
    \{(G_a,\mathrm dx^a),(G_{\bar\imath},\mathrm d\bar z^i)\}.
\end{equation}
They satisfy the descent equations
\begin{equation}
    Q\phi^{(r)}=\mathrm d'\phi^{(r-1)}.
\end{equation}
The top descendant is
\begin{equation}\label{eq:general-free-top-descendant}
    \phi^{(m+n)}
      =\bigl(G_1\dotsm G_mG_{\bar1}\dotsm G_{\bar n}\phi\bigr)
       \,\mathrm d^m x\wedge\mathrm d^n\bar z.
\end{equation}
The realisations in \eqref{eq:general-free-twist-cases} divide into two
cases.  For \(\hat d\in\{3,4\}\), there are two independent homotopies.  Their
product maps the bottom component to the auxiliary field, so there is a
nonzero convention-dependent constant \(c_Q\) such that
\begin{equation}\label{eq:general-free-second-descendant}
    \phi^{(2)}=c_QF\,
       \mathrm d^m x\wedge\mathrm d^n\bar z=-c_QP\bar\phi\,
       \mathrm d^m x\wedge\mathrm d^n\bar z.
\end{equation}
Here we used the auxiliary-field equation of motion
\(F=-P\bar\phi\) following from \eqref{eq:general-free-KW}.

For \(\hat d\in\{1,2\}\), however, \(m+n=1\), so the top descendant is
fermionic and cannot be proportional to \(F\).  In these cases choose a
basis \((\chi_{\parallel},\chi_{\perp})\) of the two components of
\(\chi\) adapted to the dimensional restriction and to \(Q\), such that
\begin{equation}\label{eq:general-free-low-dimensional-Q-action}
    G_A\phi=c_G\chi_{\parallel},\qquad
    Q\chi_{\parallel}=c_\partial\partial_A\phi,
    \qquad Q\chi_{\perp}=0,
\end{equation}
where \(A=1\) for \(\hat d=1\) and \(A=\bar1\) for \(\hat d=2\), and
\(c_G,c_\partial\) are nonzero convention-dependent constants.  Such an
adapted basis exists directly from the restricted supersymmetry algebra:
one component of \(\chi\) is the image of the unique nonzero homotopy,
whereas the other component lies in the kernel of \(Q\).  The free
quadratic superpotential gives the same-chirality fermion contraction
\begin{equation}\label{eq:general-free-fermion-contraction}
    \left\langle
       \chi_{\parallel}(\hat x,y)\chi_{\perp}(0,0)
    \right\rangle
       =c_\chi P G_P(\hat x,y),
\end{equation}
with \(c_\chi\neq0\); equivalently, this is the off-diagonal entry of the
inverse quadratic fermion operator.  Consequently the first-descendant
OPE with the \(Q\)-closed fermion \(\chi_{\perp}\) supplies the required
top-degree coefficient,
\begin{equation}\label{eq:general-free-first-descendant-coefficient}
    \left\langle
       \phi^{(1)}(\hat x,y)\chi_{\perp}(0,0)
    \right\rangle
       =c'_Q P G_P(\hat x,y)e^A,
       \qquad \hat d\in\{1,2\},
\end{equation}
for some nonzero constant \(c'_Q\).  Thus the low-dimensional cases use a
fermionic OPE channel rather than incorrectly identifying the first
descendant with the bosonic auxiliary field.

\subsection{The sheaf-cohomology class of the Green function}
The Green's function \(G_P\) for the quadratic action \eqref{eq:general-free-bosonic-action} appears in the OPE of the untwisted theory as
\begin{equation}
    \phi(\hat x,y)\bar\phi(0,0)=G_P(\hat x,y)\mathbb 1+\dotsb.
\end{equation}
The Fourier transform of \eqref{eq:general-free-bosonic-action} yields
\begin{equation}\label{eq:general-free-Green-function}
    G_P(\hat x,y)
      =\int_{\mathbb R^{\hat d}}
         \frac{\mathrm d^{\hat d}\hat k}{(2\pi)^{\hat d}}
       \int_{\mathbb R^{\bar d}}
         \frac{\mathrm d^{\bar d}k}{(2\pi)^{\bar d}}
       \frac{\exp(\mathrm i\hat k\cdot\hat x+\mathrm ik\cdot y)}
       {\|\hat k\|^2+P(\lVert k\rVert^2)^2},
\end{equation}
where \(P(\lVert k\rVert^2)\) denotes the Fourier symbol of \(P(-\Delta_y)\).
Combining \eqref{eq:general-free-second-descendant} and
\eqref{eq:general-free-first-descendant-coefficient}, define
\begin{equation}\label{eq:general-free-OPE-coefficient}
    \alpha_P(\hat x,y)\coloneqq
    \begin{cases}
      \left\langle
        \phi^{(1)}(\hat x,y)\chi_{\perp}(0,0)
      \right\rangle,
        & \text{if \(\hat d\in\{1,2\}\)}\\
      \left\langle
        \phi^{(2)}(\hat x,y)\phi(0,0)
      \right\rangle,
        & \text{if \(\hat d\in\{3,4\}\)}.
    \end{cases}
\end{equation}
After absorbing various coefficients into a nonzero constant \(\tilde c_Q\), this has the uniform form
\begin{equation}\label{eq:general-free-uniform-coefficient}
    \alpha_P(\hat x,y)=\tilde c_Q P G_P(\hat x,y)\,
       \mathrm d^m x\wedge\mathrm d^n\bar z.
\end{equation}
Let
\begin{equation}
    X=\bigl(\mathbb R^m\times\mathbb C^n
       \times\mathbb R^{\bar d}\bigr)\setminus\{0\}
\end{equation}
be the topological--holomorphic--ordinary spacetime minus the origin.
The form \(\alpha_P\) has top degree in the de~Rham--Dolbeault
complex
\begin{equation}
    \left(
      \mathcal C^\infty(X)\otimes
      \bigwedge^\bullet\langle
        \mathrm dx^1,\ldots,\mathrm dx^m,
        \mathrm d\bar z^1,\ldots,\mathrm d\bar z^n
      \rangle,
      \mathrm d_x+\bar\partial_z
    \right).
\end{equation}
It is therefore \((\mathrm d_x+\bar\partial_z)\)-closed.
Furthermore, it can be seen from the Fourier representation that it is smooth away from the origin of the total spacetime.  Thus, it
defines a class
\begin{equation}\label{eq:general-free-cohomology-class}
    [\alpha_P]\in\operatorname H^{m+n}(X,\mathcal O)
      \cong \tilde Y_{n,\bar d}/Y_{n,\bar d}.
\end{equation}

We now argue that this cohomology class is nontrivial.  Since the hatted
fibres are noncompact, fibre integration is not a priori a map on the full
sheaf-cohomology group.  For the Green-function representative
\(\alpha_P\), however, it is well defined as a tempered distribution in
\(y\).  Wedge \eqref{eq:general-free-uniform-coefficient} with
\(\mathrm d^nz=\mathrm dz^1\wedge\dotsm\wedge\mathrm dz^n\) and integrate
over the hatted coordinates (this is equivalent to a pushforward of tempered distributions).  In momentum space this restricts the
Fourier transform to \(\hat k=0\), and hence
\begin{equation}\label{eq:general-free-transverse-Green-function}
\begin{aligned}
    \int_{\mathbb R^m\times\mathbb C^n}
       \alpha_P\wedge\mathrm d^nz
      &=C_Q\int_{\mathbb R^{\bar d}}
         \frac{\mathrm d^{\bar d}k}{(2\pi)^{\bar d}}
         \frac{\exp(\mathrm ik\cdot y)}{P(\lVert k\rVert^2)}\\
      &=C_QP^{-1}(y),
\end{aligned}
\end{equation}
for a nonzero constant \(C_Q\). Thus, modulo analytical subtleties discussed in \cref{app:pushforward-obstruction},
\eqref{eq:general-free-uniform-coefficient} represents a nonzero element of
sheaf cohomology as predicted by \cref{thm:main}.

\section{Discussion}\label{sec:discussion}
The discussion above was limited to the two-point OPE (or, equivalently, three-point correlation functions). It is straightforward to generalise the above to the multi-point case at finite separation by considering analogues of the sheaf \(\mathcal O\) on higher-point configuration spaces (since \((\mathbb R^m\times\mathbb C^n\times\mathbb R^p)\setminus\{0\}\) can be thought of as a factor of the two-point configuration space
\begin{equation*}\left\{(x,y)\middle|x,y\in \mathbb R^m\times\mathbb C^n\times\mathbb R^p,\;x\ne y\right\}
=
\left((\mathbb R^m\times\mathbb C^n\times\mathbb R^p)\setminus\{0\}\right)
\times(\mathbb R^m\times\mathbb C^n\times\mathbb R^p).
\end{equation*}
However, to properly discuss the multi-point case, one should understand higher analogues of the Borcherds identity, i.e.\ how to compose lower-point correlation functions into higher-point correlation functions. A natural language for doing so would involve chiral algebras \cite{zbMATH02121180,zbMATH07376281,Gaitsgory:1998uq} or factorisation algebras \cite{Costello:2023knl,Amabel:2023hzs,Amabel:2019yrn}. For recent discussion, see \cite{2024arXiv240816787M, Griffin:2025add}.

An interesting tangent is that the sheaf cohomology that we have constructed admits the structure of a \(C_\infty\)-algebra (commutative \(A_\infty\)-algebra), since it is the sheaf cohomology of a sheaf of algebras, not just of vector spaces. That is, the cohomology, in addition to being a graded-commutative algebra, also has products of higher arity, such as a ternary product \(m_3\), a quaternary product \(m_4\), etc. One can keep track of such higher structures by means of the Thom--Sullivan model, as done for instance in \cite{Alfonsi:2024qdr}. Presumably such structures may enter into the composition of OPEs in higher analogues of the Borcherds identity.

One can generalise the considerations in this paper even further to more general patterns of topological and holomorphic spacetime directions. Given a smooth manifold \(M\), one can consider a complex subbundle \(D\subset \mathrm T^{\mathbb C}M\) of the complexified tangent bundle \(\mathrm T^{\mathbb C}M\) whose sections are closed under the Lie bracket of vector fields, and consider field theories whose correlation functions are invariant under translations along \(D\). This may be thought of as a complex analogue of a regular foliation (according to Frobenius' theorem).
If one further relaxes the condition that \(D\) be a subbundle to just requiring it to be a suitable subsheaf, then one may obtain singular foliations, and at singular points of the foliation one can obtain richer behaviour of OPE coefficients than discussed herein. Constructing examples of such theories, however, seems difficult.

\section*{Acknowledgements}
The authors thank Laura Olivia Felder\textsuperscript{\orcidlink{0009-0006-9116-6007}}, Leron Borsten\textsuperscript{\orcidlink{0000-0001-9008-7725}}, and Charles Alastair Stephen Young\textsuperscript{\orcidlink{0000-0002-7490-1122}} for helpful discussion.

\appendix
\section{Proof of the main theorem}\label{app:appendix}
Recall that we are working on a spacetime \(\mathbb R^m\times\mathbb C^n\times\mathbb R^p\) with \(m\) topological directions, \(n\) holomorphic directions (and \(n\) anti-holomorphic directions), and \(p\) ordinary directions. We will be working with sheaves of complex topological vector spaces whose sections are Fréchet nuclear spaces; when we write \(\otimes\), we mean the topological tensor product (of Fréchet nuclear sheaves), for which the Künneth formula holds \cite{zbMATH03304714}, cf.\ \cite{F2005}. The symbols \(Z_{n,p}\), \(\tilde Z_{n,p}\), \(Y_{n,p}\) and \(\tilde Y_{n,p}\) are defined in \eqref{eq:function_spaces}.

The main technical tool we use is the Čech-to-derived spectral sequence --- or, rather, the special case involving a cover of two open sets only, in which case it reduces to a Mayer--Vietoris exact sequence for sheaf cohomology.

\begin{lemma}\label{lem:holomorphic_raviolo_computation}
Consider the case where \(p=0\), i.e.\ \((\mathbb R^m\times\mathbb C^n)\setminus\{0\}\) equipped with the sheaf \(\mathcal O\) of complex-valued smooth functions that are constant along \(\mathbb R^m\) and holomorphic along \(\mathbb C^n\). Then, as graded complex vector spaces, its sheaf cohomology is
\begin{equation}
    \operatorname H((\mathbb R^m\times\mathbb C^n)\setminus\{0\},\mathcal O)
    \cong\begin{cases}
        Z_{n,0}\oplus Y_{n,0}[1-m-n] & \text{if \(m+n>0\)} \\
        0 & \text{if \(m=n=0\)}.
    \end{cases}
\end{equation}
\end{lemma}

\begin{proof}
First, the case \(m=n=0\) is trivial. Second, the case \(m\ge1,n=0\) reduces to the singular cohomology of \(\mathbb R^m\setminus\{0\}\), which is \(\mathbb C\oplus\mathbb C[1-m]\). Likewise, the case \(m=0,n\ge2\) follows from lemma~\ref{lem:sheaf-coho-result} and the definition of \(Y_{n,0}\). Instead, for \(m=0,n=1\), Laurent decomposition gives \(\mathcal O(\mathbb C^\times)\cong Z_{1,0}\oplus Y_{1,0}\). Thus, from now on, let us assume \(m,n\ge1\).
We may cover \(X_{m,n}\coloneqq(\mathbb R^m\times\mathbb C^n)\setminus\{0\}\) by \(U\coloneqq(\mathbb R^m\setminus\{0\})\times\mathbb C^n\) and \(V\coloneqq\mathbb R^m\times(\mathbb C^n\setminus\{0\})\). Then, by the K\"unneth theorem, the sheaf cohomology on each of the two patches and their intersection is clearly
\begin{align}
    \operatorname H(U,\mathcal O)&\cong Z_{n,0}\oplus Z_{n,0}[1-m],\\
    \operatorname H(V,\mathcal O)&\cong Z_{n,0}\oplus Y_{n,0}[1-n],\\
    \operatorname H(U\cap V,\mathcal O)&\cong
    Z_{n,0}\oplus Y_{n,0}[1-n]\oplus Z_{n,0}[1-m]\oplus Y_{n,0}[2-m-n].
\end{align}
Here, we choose the splittings into constant and reduced cohomology in the real factor and, when \(n=1\), into regular and negative Laurent parts in the complex factor. Thus restriction from each unpunctured factor lands in its regular summand, including when \(m=1\) or \(n=1\).

It is immediate that
\begin{equation}
    \operatorname H^0((\mathbb R^m\times\mathbb C^n)\setminus\{0\},\mathcal O)\cong Z_{n,0}.
\end{equation}
The map
\begin{equation}
    \operatorname H(U,\mathcal O)\oplus\operatorname H(V,\mathcal O)
    \longrightarrow\operatorname H(U\cap V,\mathcal O)
\end{equation}
in the Mayer--Vietoris long exact sequence consists of a direct sum of
\begin{equation}
    Z_{n,0}\oplus Z_{n,0}\xrightarrow{\ (a,b)\mapsto a-b\ }Z_{n,0}
\end{equation}
and
\begin{equation}
    Y_{n,0}[1-n]\xrightarrow{-\operatorname{id}}Y_{n,0}[1-n]
\end{equation}
and
\begin{equation}
    Z_{n,0}[1-m]\xrightarrow{\operatorname{id}}Z_{n,0}[1-m]
\end{equation}
and
\begin{equation}
    0\longrightarrow Y_{n,0}[2-m-n].
\end{equation}
Its kernel is the diagonal copy of \(Z_{n,0}\) in degree zero, and its cokernel is \(Y_{n,0}\) in degree \(m+n-2\). The connecting homomorphism then implies that \(\operatorname H^{m+n-1}(X_{m,n},\mathcal O)\cong Y_{n,0}\), whilst all other cohomology groups vanish. This then shows the claimed result.
\end{proof}

\begin{proof}[Proof of \cref{thm:main}]
The case \(p=0\) is precisely lemma~\ref{lem:holomorphic_raviolo_computation}. The case \(m=n=0,p\ge1\) is also trivial, as the sheaf is \(\mathcal C^\infty\) and its cohomology is just \(\tilde Z_{0,p}=\mathcal C^\infty(\mathbb R^p\setminus\{0\},\mathbb C)\cong Z_{0,p}\oplus (\tilde{Y}_{0,p}/Y_{0,p})\) in degree zero.
Thus, we assume \(m+n>0\) and \(p\ge1\). We may cover \(X_{m,n,p}\coloneqq(\mathbb R^m\times\mathbb C^n\times\mathbb R^p)\setminus\{0\}\) by \(U\coloneqq((\mathbb R^m\times\mathbb C^n)\setminus\{0\})\times\mathbb R^p\) and \(V\coloneqq\mathbb R^m\times\mathbb C^n\times(\mathbb R^p\setminus\{0\})\). We know the sheaf cohomology on each of the two patches and their intersection by lemma~\ref{lem:holomorphic_raviolo_computation} and K\"unneth theorem:
\begin{align}
    \operatorname H(U,\mathcal O)&\cong Z_{n,p}\oplus Y_{n,p}[1-m-n],\\
    \operatorname H(V,\mathcal O)&\cong\tilde Z_{n,p},\\
    \operatorname H(U\cap V,\mathcal O)&\cong\tilde Z_{n,p}\oplus\tilde Y_{n,p}[1-m-n],
\end{align}
where all tensor products are completed, as above. 
The cohomology of \(X_{m,n,p}\) can now be computed using the Mayer--Vietoris long exact sequence for sheaf cohomology.
First, consider the restriction maps \(r_Z\colon Z_{n,p}\to\tilde Z_{n,p}\) and \(r_Y\colon Y_{n,p}\to\tilde Y_{n,p}\), which are injective.
The map given by the difference of the restrictions is a direct sum of
\begin{equation}\label{eq:general_computation_diag1}
    \operatorname H^0(X_{m,n,p},\mathcal O)\rightarrow Z_{n,p}\oplus\tilde Z_{n,p}\xrightarrow{\; (a,b)\mapsto r_Z(a)-b\; }\tilde Z_{n,p}\rightarrow 0
\end{equation}
in degree zero and
\begin{equation}\label{eq:general_computation_diag2}
    \operatorname H^{m+n-1}(X_{m,n,p},\mathcal O)\rightarrow Y_{n,p}\xrightarrow{\; r_Y\; }\tilde Y_{n,p}\rightarrow\operatorname H^{m+n}(X_{m,n,p},\mathcal O)
\end{equation}
in degree \(m+n-1\).
For \(f^k\) the degree-\(k\) part of this map, exactness gives
\begin{equation}
    0\longrightarrow\operatorname{coker}f^{k-1}
    \longrightarrow\operatorname H^k(X_{m,n,p},\mathcal O) 
    \longrightarrow\ker f^k\longrightarrow0,
\end{equation}
The map \eqref{eq:general_computation_diag1} is surjective with kernel \(\{(a,r_Z(a)):a\in Z\}\cong Z\) in degree zero, so \(\operatorname H^0(X_{m,n,p},\mathcal O) \cong Z_{n,p}\). The map \eqref{eq:general_computation_diag2} is injective with cokernel \(\tilde Y_{n,p}/Y_{n,p}\) in degree \(m+n-1\), thus \(\operatorname H^{m+n}(X_{m,n,p},\mathcal O) \cong \tilde Y_{n,p}/Y_{n,p}\), and all other cohomology groups vanish, as claimed.
\end{proof}

\section{A cutoff argument for the Green-function class}\label{app:pushforward-obstruction}

In this appendix we justify the obstruction used in
\cref{eq:general-free-transverse-Green-function}. 
The point requiring care is that exactness of \(\alpha_P\) supplies no
decay for a primitive at infinity. We therefore use a compactly supported cutoff in the hatted variables.
\Cref{lem:green-pushforward-obstruction} shows that exactness of \(\alpha_P\)
would force the restriction of its pushforward to
\(\mathbb R^{\bar d}\setminus\{0\}\) to extend smoothly through \(y=0\),
whereas \cref{lem:transverse-green-nonextendability} shows that the
restriction of \(P^{-1}(y)\) has no such extension. 
Thus, applying the two lemmas to
\eqref{eq:general-free-transverse-Green-function} shows that
\([\alpha_P]\neq0\).

Put
\(H=\mathbb R^m\times\mathbb C^n\), let \(q=m+n\), and write
\(d'=\mathrm d_x+\bar\partial_z\).  We orient the underlying real vector
space of \(H\) so that wedging a partial top form with
\(\mathrm d^nz\) gives an ordinary top form.

\begin{lemma}\label{lem:green-pushforward-obstruction}
Let \(P=P(-\Delta_y)\) be as in \cref{sec:theory-examples}.  Let \(\alpha_P\) be the partial \(q\)-form in
\eqref{eq:general-free-uniform-coefficient}.  Suppose that there is a smooth
partial \((q-1)\)-form \(\beta\) on
\(X=(H\times\mathbb R^{\bar d})\setminus\{0\}\) such that
\(\alpha_P=d'\beta\).  Then the restriction to \(\mathbb R^{\bar d}\setminus\{0\}\) of
\[
 y\mapsto\int_H\alpha_P(\hat x,y)\wedge\mathrm d^nz
\]
extends to a smooth function in a neighbourhood of \(y=0\).
\end{lemma}

\begin{proof}
Choose \(R>0\) and a smooth real-valued function \(\chi\) on \(H\) which is
identically one on the ball \(B_R\) and supported in \(B_{2R}\).  Split the
pushforward as
\begin{equation}\label{eq:green-pushforward-cutoff-split}
 \int_H\alpha_P\wedge\mathrm d^nz
 =\int_H\chi\alpha_P\wedge\mathrm d^nz
  +\int_H(1-\chi)\alpha_P\wedge\mathrm d^nz.
\end{equation}
For the first term, use \(\alpha_P=d'\beta\).  Since
\(\chi\beta\wedge\mathrm d^nz\) has compact support in the hatted
variables, Stokes' theorem gives
\begin{equation}\label{eq:green-pushforward-cutoff-stokes}
 \int_H\chi\alpha_P\wedge\mathrm d^nz
 =\pm\int_H d'\chi\wedge\beta\wedge\mathrm d^nz.
\end{equation}
(The sign depends only on the degree convention and is irrelevant here.)
The support of \(d'\chi\) is contained in the annulus
\(B_{2R}\setminus B_R\).  This annulus stays a positive distance from the
deleted point for every \(y\) near zero.  The form \(\beta\) is therefore
jointly smooth on a neighbourhood of
\(\operatorname{supp}(d'\chi)\times\{0\}\), and differentiation under the
compact integral in \eqref{eq:green-pushforward-cutoff-stokes} shows that
the first term is smooth in \(y\) near zero.

For the second term in
\eqref{eq:green-pushforward-cutoff-split},
its Fourier transform in \(y\) only into the variable \(k\) (but leaving \(\hat x\) intact) is dominated by polynomial factors times \(\exp(-P(\|k\|^2)|\hat x|\), hence integrable. Differentiation under the integral then proves smoothness of this second term.
Both terms in \eqref{eq:green-pushforward-cutoff-split} are thus smooth there.
\end{proof}

\begin{lemma}\label{lem:transverse-green-nonextendability}
The restriction of the tempered distribution \(P^{-1}(y)\) to
\(\mathbb R^{\bar d}\setminus\{0\}\) does not admit a smooth extension
through the origin.
\end{lemma}

\begin{proof}
Suppose, for contradiction, that this restriction admits a smooth extension
\(f\) through \(y=0\).
Set \(S=P^{-1}-f\), so that \(S\) is supported at the origin.
Moreover, \(P(-\Delta_y)f=0\) away from the origin.  Since \(P(-\Delta_y)f\) is
smooth, it follows that \(P(-\Delta_y)f=0\) everywhere.  Hence
\[
    P(-\Delta_y)S=\delta_0.
\]
Every distribution supported at the origin is a finite linear combination
of derivatives of \(\delta_0\), so its Fourier transform is a polynomial
\(q(k)\). Taking Fourier transforms gives
\[
    P(\lVert k\rVert^2)q(k)=1,
\]
which is impossible because \(P\) is nonconstant and \(q\) is a polynomial.
\end{proof}

\bibliographystyle{unsrturl}
\newcommand\cyrillic[1]{\fontfamily{Domitian-TOsF}\selectfont \foreignlanguage{russian}{#1}}
\newcommand\ukrainian[1]{\fontfamily{Domitian-TOsF}\selectfont \foreignlanguage{ukrainian}{#1}}
\bibliography{biblio}
\end{document}